\pdfoutput=1
\documentclass[aps, prl, preprintnumbers, notitlepage, superscriptaddress, twocolumn, nofootinbib]{revtex4-1}
\usepackage{graphicx}  
\usepackage{dcolumn}   
\usepackage{bm}        
\usepackage{amsmath, amssymb, amsthm, mathrsfs, amsfonts, dsfont}
\usepackage{mathtools}
\usepackage{mathpazo}

\newcommand{\nc}{\newcommand}
\nc{\rnc}{\renewcommand}
\nc\mnb[1]{\medskip\noindent{\bf #1}}

\newcommand{\ket}[1]{\left| #1 \right>} 
\newcommand{\vect}[1]{\bm{#1}}

\newcommand{\smat}[4]{\left(\begin{smallmatrix} #1 & #2 \\ #3 & #4 \end{smallmatrix}\right)}

\newcommand{\z}{Z}
\newcommand{\x}{X}
\newcommand{\y}{Y}
\newcommand{\R}{\mathbb{R}}
\newcommand{\Prob}{\mathbb{P}}
\newcommand{\LieA}[1]{\mathfrak{#1}}
\newcommand{\LieG}[1]{\mathcal{#1}}

\newcommand{\sr}{shrinking region}
\newcommand{\s}{ }
\newcommand{\jr}{jumping region}
\newcommand{\jrs}{jumping regions}
\newcommand{\rc}{random conjugation}
\newcommand{\rcs}{random conjugations}

\newcommand{\eq}[1]{eq.~\eqref{#1}}

\newcommand{\iden}{\mathds{1}}
\newcommand{\supp}{Appendix}

\newcommand{\comb}[2]{\begin{pmatrix} #1 \\ #2 \end{pmatrix}}
\newcommand{\supdesc}[1]{\textit{#1}\vspace{2ex}}

\newcommand{\eqs}[1]{\eq{#1}}
\newcommand{\eql}[1]{\eq{#1}}

\newtheorem{thm}{Theorem}

\newtheorem{lem}[thm]{Lemma}

\DeclareMathOperator{\poly}{poly}
\DeclareMathOperator{\polylog}{polylog}

\DeclareMathOperator{\tr}{Tr}

\DeclareMathOperator{\vol}{vol}

\usepackage{hyperref}
\hypersetup{colorlinks, citecolor=magenta, filecolor=blue, linkcolor=blue, urlcolor=green}

\begin{document}


\title{Universal Refocusing of Systematic Quantum Noise}

\author{Imdad S. B. ~Sardharwalla} \affiliation{Department of Applied Mathematics and Theoretical Physics, University of Cambridge, Cambridge, CB3 0WA, U.K.}
\author{Toby S. ~Cubitt} \affiliation{Department of Computer Science, University College London, Gower Street, London WC1E 6BT, U.K.}
\author{Aram W. ~Harrow} \affiliation{Center for Theoretical Physics, MIT, U.S.A.}
\author{Noah ~Linden} \affiliation{School of Mathematics, University of Bristol, Bristol, BS8 1TW, U.K.}

\vskip 0.25cm

\begin{abstract}
  Refocusing of a quantum system in NMR and quantum information processing can be achieved by application of short pulses according to the methods of spin echo and dynamical decoupling.
  However, these methods are strongly limited by the requirement that the evolution of the system between pulses be suitably small.
  Here we show how refocusing may be achieved for arbitrary (but time-independent) evolution of the system between pulses.
  We first illustrate the procedure with one-qubit systems, and then generalize to $d$-dimensional quantum systems.
  We also give an application of this result to quantum computation, proving a new version of the Solovay-Kitaev theorem that does not require inverse gates.
\end{abstract}

\maketitle

An isolated quantum system will evolve in time according to its inherent Hamiltonian.
This can often be an undesirable effect that needs to be corrected.
Within the field of nuclear magnetic resonance (NMR), a technique
known as spin echo is often employed to correct for some kinds of
evolution by applying a particular radio-frequency pulse to the system at a certain time that causes the state of the system to `refocus' \cite{hahn1950spin,freeman1998spin}.

It is also an issue very commonly encountered within quantum
information processing
\cite{alvarez2012iterative,divincenzo2000physical,khodjasteh2005fault,viola1999dynamical,witzel2007concatenated,yang2011preserving,zhao2014dynamical,kuo2011quadratic},
where an unwanted always-on evolution leads to a coupling between two
initially isolated  systems.
One method of dealing with this, known as dynamical decoupling, is an extension of refocusing in spin echo, and involves applying several `control pulses' to the combined system over a period of time to dynamically eliminate the coupling~\cite{kuo2011quadratic}.
In addition, the sequence and timing of the control pulses is independent of both systems.

The main difficulty in using dynamical decoupling methods in general is that the method requires the joint-system evolution between pulses to be small---the larger the Hamiltonian that produces the coupling, the smaller the time interval between the pulses must be for the dynamical decoupling method to be effective, and hence the more frequent the pulses must be.
This could pose a problem in some systems with strong coupling, or in systems where pulses cannot be applied as frequently.
Thus it is clear that at some point, the dynamical decoupling methods must break down.

In Quantum Computing, this idea recasts itself as a different problem, viz.\ the problem of inverting an unknown black-box unitary operation $U$, given access to as many uses of the black-box as necessary.
If additional ancilla systems are available, this can in principle be achieved by performing full process tomography of the operator.
The results of~\cite{sheridan2009approximating} give another ancilla-assisted method for achieving this, without requiring full tomography, in which the number of control unitaries scales as a polynomial in $1/\epsilon$, where $\epsilon$ is the error in the output. 
However, in many practical scenarios, ancillas are not available.
Even when they are, carefully engineering complex interactions between ancilla systems and the system to be refocused is typically difficult or infeasible.
We will work in a much more restrictive setup, in which all control unitaries are required to act on a single system with the state space of $U$.
For example, if $U$ is a one-qubit operator, we only allow operations to be performed on that one qubit.

In this letter, we derive a universal procedure to refocus any unitary $U$ to arbitrary accuracy.
We find a sequence of unitary operations $\{R_1, \dots, R_n\}$, independent of $U$, such that
\begin{equation}
  R_1 U R_2 U \cdots U R_n U \approx \iden.\label{eq:format}
\end{equation}
More precisely, $\|R_1 U R_2 U \cdots U R_n U - \iden\| \leq \epsilon$, where the number $n$ of control unitaries $R$ only needs to scale as $n = O(\log^2(1/\epsilon))$.
Since the procedure works for arbitrary $U$, it is able to refocus completely unknown time-independent unitary dynamics---or equivalently, arbitrary, unknown, fixed Hamiltonian dynamics of any strength.

\mnb{Efficient gate approximation without inverses.}
An application of our refocusing result to quantum computation is to extend one of the central results in quantum compiling---the Solovay-Kitaev theorem~\cite{nielsen2010quantum, kitaev2002classical}---to the case when inverse gates are not included.
Informally, the original Solovay-Kitaev result proves that a universal quantum gate set \emph{that includes inverse gates} can simulate any other universal gate set to arbitrary precision $\epsilon$, with at most $\log^{3+o(1)}(1/\epsilon)$ overhead.
This is fundamental to the theory of quantum circuits and to practical quantum computation, as it shows that any universal gate set can simulate any other with low overhead.
In a circuit of size $L$ we can think of $\epsilon$ as $O(1/L)$, so changing from one universal gate set to another would increase the number of gates to at most $L\log^{3+o(1)}(L)$.
However, when inverse gates are not included, all known variants of the Solovay-Kitaev theorem \cite{dawson2005solovay,kitaev2002classical} fail.  The only previously known method of approximating the inverse of a gate $U$ was to wait until a member of the sequence $U,U^2,U^3,\ldots$ approximated $U^{-1}$, which in general required overhead $1/\epsilon^{d^2-1}$, i.e.~$1/\epsilon^3$ for qubits.
Thus, a circuit of size $L$ would turn into $\poly(L)$ gates, which is a large enough overhead to overwhelm the polynomial speedup from algorithms such as Grover's.
By using our refocusing result to efficiently approximate inverse gates, we obtain a new inverse-free version of the Solovay-Kitaev theorem: Any universal quantum gate set that includes the Pauli operators (or Weyl operators for qudits; see below) can simulate any other universal gate set to arbitrary precision $\epsilon$, with at most $\poly\log(1/\epsilon)$ overhead.
(See \supp~for proof details.)

In the remainder of the paper we describe the refocusing procedure
for a one-qubit system, and then derive the general case of $d$-dimensional systems.

\textbf{One-Qubit Unitary Noise}.
We describe here the procedure to eliminate systematic noise on one-qubit systems.
Any unitary operation $U\in\LieG{SU}(2)$ may be written in the form $U=e^{-iH}$, where the Hamiltonian $H$ is of the form $H=\vect{h}\cdot\vect{\sigma}$, where $\vect{\sigma}=(\x,\y,\z)^T$ is the vector of Pauli matrices, and $\vect{h}\in\R^3$.

We introduce the function
\begin{equation*}
	f(U)\coloneqq \x U \x \y U \y \z U \z U
\end{equation*}
which can be seen to give $\iden$ to first order in $H$ when expanded as a power series.
Thus, we expect that for $U$ within a certain distance of $\iden$, the recursive application of $f$ will reduce this distance.
This forms the basis of concatenated dynamical decoupling \cite{yang2011preserving,witzel2007concatenated,kuo2011quadratic,khodjasteh2005fault,alvarez2012iterative}, and one of the stages of our procedure.

Outside of this region, $f$ does not necessarily reduce this distance; in fact, $f$ has several fixed points and cycles.
For example, the unitary operator $\frac{1-i}{2} \smat{1}{i}{-1}{i}$ is a fixed point, and $(\frac{1-i}{2} \smat{i}{i}{-1}{1}, \frac{1-i}{2} \smat{1}{-1}{i}{i})$ is a two-cycle.
This is a key motivation for developing a randomised (rather than deterministic) protocol for refocusing.

Note that $f$ can be expressed in the form of \eq{eq:format} as $f(U)=-\x U \z U \x U \z U$.

The analysis for the one-qubit case can be computed explicitly, and we do so in the following three stages:

\begin{enumerate}
\item
  In terms of a chosen measure of distance, we lower bound the size of the neighbourhood of $\iden$ for which an application of $f$ reduces the distance to $\iden$.
  We shall call this the \textit{\sr}.
  This is the crux of concatenated dynamical decoupling, which can only be applied within this region.

\item
  We find other points in $\LieG{SU}(2)$ that are mapped exactly to $\iden$ under a single application of $f$, and hence (by continuity of $f$) determine regions that are mapped into the \sr.
  We call these \textit{\jrs}.

\item
  We apply certain random operations to our unitary and lower bound the probability of moving it into one of the \jrs.
  We call these \textit{\rcs}.
\end{enumerate}

\textit{Bounding the \sr}.---Any $U\in\LieG{SU}(2)$ can be written in the form (see \supp)
\begin{equation}
  \begin{aligned}
    &U=a\iden+ib\x+ic\y+id\z,\\
    &a^2+b^2+c^2+d^2=1, \hspace{1.5em}a,b,c,d\in \R.	\label{eq:sumsq}
  \end{aligned}
\end{equation}

Using the Hilbert-Schmidt norm $\|A\|=\frac{1}{2}\sqrt{\tr(A^\dagger A)}$, we define the
\emph{distance} between $U$ and $\iden$ to be $\varepsilon_0 \coloneqq
\|U-\iden\|=\sqrt{1-a}$.
A straightforward matrix multiplication then tells us that the distance between $f(U)$ and $\iden$ is $\varepsilon_1 \coloneqq \sqrt{8}|bd|$.
Now
\begin{equation}
  \varepsilon_1=\sqrt{8}|bd| \leq \sqrt{2}(b^2+d^2) \leq \sqrt{2}(1-a^2) \leq \sqrt{8}\varepsilon_0^2.\label{eq:1q-sr1}
\end{equation}
where the second inequality follows from \eq{eq:sumsq}.

If $\varepsilon_m$ is the distance from $\iden$ after $m$ applications of $f$, then repeated application of \eq{eq:1q-sr1} implies that $\varepsilon_m\leq \sqrt{8}^{2^m-1}\varepsilon_0^{2^m}$.
Choosing $\varepsilon_0\leq 1/4$ gives us doubly-exponential
convergence towards $\iden$ as $m$ increases, that is,
$\varepsilon_m\leq \frac{1}{\sqrt{8}} 2^{-2^{m-1}}$.

We thus define the \sr\s to be $1-a = \varepsilon_0^2 \leq 1/16$.
This is represented in Figure~\ref{fig:sphere} as region $A$.

\textit{Bounding the \jrs}.---We saw previously that
$\varepsilon_1=\sqrt{8}|bd|$. 
To simply ensure that $f(U)$ be inside the \sr, we require that
$\varepsilon_1\leq 1/4$.
 Denote the ``jumping'' region by $J \equiv f^{-1}(A)$, and observe
 that $J$ is the set of $U$ with $|bd| \leq 1/\sqrt{128}$; see Figure \ref{fig:sphere}.

\begin{figure}
  \includegraphics[width=6cm]{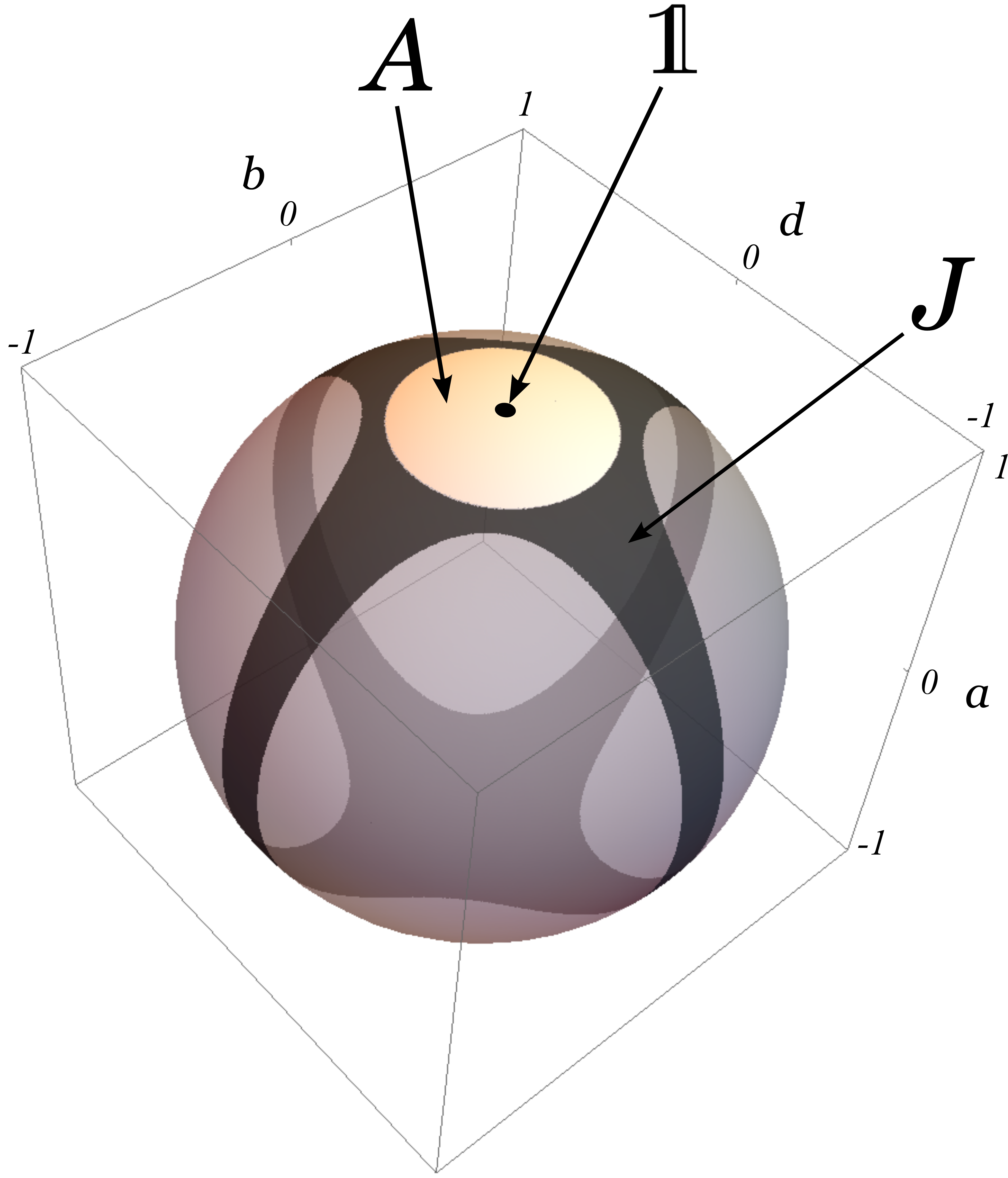}
  \caption{Universal refocusing for $U\in\LieG{SU}(2)$.
    For illustration we set $c=0$ in \protect \eq{eq:sumsq} so that the surface of the sphere represents the remaining part of $\LieG{SU}(2)$.
    $A$ represents the \sr, with $U=\iden$ marked at its center point.
    $J \equiv f^{-1}(A)$ is the \jr, for which $|bd| \leq 1/\sqrt{128}$.
    The action of a \rc\s $R=\vect{r}\cdot\vect{\sigma}$ (where, for
    this illustration, $\vect{r}=(r_1,0,r_3)$) is to reflect the
    sphere in a plane along the $a$ axis containing $\iden$, leaving the distance to $\iden$
    invariant.}
  \label{fig:sphere}
\end{figure}

\textit{Bounding the probability of landing in a \jr\s after applying a \rc}.---We now write $U$ in the form $U=a\iden+i(\vect{u}\cdot\vect{\sigma})$, where $\vect{u}=(b,c,d)^T$ and $\vect{\sigma}=(\x,\y,\z)^T$.
The operation we apply is conjugation by an operator $R=\vect{r}\cdot\vect{\sigma}$, where $\vect{r}$ is a real unit vector, and $R$ is unitary.
Then
\begin{equation*}
  U^\prime=RUR^\dagger=a\iden+i\vect{u^\prime}\cdot\vect{\sigma}.
\end{equation*}
where $\vect{u^\prime}=[2(\vect{r}\cdot \vect{u})\vect{r}-\vect{u}]=(b^\prime, c^\prime, d^\prime)^T$.
This transformation has two important properties:

\begin{itemize}
\item
  The distance from $\iden$ is invariant.
  This ensures that the unitary can never leave the \sr\s once inside it;

\item
  $\vect{u^\prime}$ is the rotation of $\vect{u}$ by $\pi$ about the vector $\vect{r}$.
  Thus choosing $\vect{r}$ to point in a uniformly random direction (according to the spherical measure on $S^2$) ensures that $\vect{u^\prime}$ also points in a similarly uniformly random direction (with $|\vect{u^\prime}|=|\vect{u}|$).
  In Figure \ref{fig:sphere}, this would be represented by a reflection of the sphere in a vertical plane.
\end{itemize}

We now lower bound the probability that $U^\prime$ is in a \jr.
To do so, we write $\vect{u^\prime}$ in spherical coordinates:
$\vect{u^\prime}=(b^\prime, c^\prime,
d^\prime)^T_{\text{cart}}=(|\vect{u}|,\theta,\phi)^T_{\text{sph}}$. 
The jumping region $J$ corresponds to the unitaries with 
\begin{align}
  |\vect{u}|^2 |\cos(\theta)\sin(\theta)\cos(\phi)| \leq \frac{1}{\sqrt{128}}.
\end{align}
Recall that $\theta,\phi$ are drawn uniformly at random from the
sphere, while $|\vect{u}|$ depends on $U$.  To eliminate this
dependence we can bound
\begin{align}
  \Prob[U'\in J] & \geq \Prob\left[|\cos(\theta)\sin(\theta)\cos(\phi)| \leq 1/\sqrt{128}\right] \nonumber \\
  & \approx 0.271\ldots \nonumber
\end{align}
The constant $0.271\ldots$ can be obtained by numerical integration,
and for notational convenience we will simply use $\Prob[U'\in J] \geq
1/4$.

\vspace{2ex}
We now introduce the function $g(U)=(\vect{r}\cdot\vect{\sigma})U(\vect{r}\cdot\vect{\sigma})^\dagger$, where each application of $g$ chooses a unit direction vector $\vect{r}$ uniformly at random according to the spherical measure on $S^2$.
Consider $(f\circ g)^{\circ l}$, i.e.\
$f$ and $g$ composed $l$ times.
In order to enter a \jr\s with probability $\geq 1-\eta$ we require
\begin{equation*}
  l\geq \frac{\log_2(1/\eta)}{\log_2(4/3)}.
\end{equation*}

Once in the \sr, we require a further $m$ steps to get within
$\epsilon := \epsilon_{l+m}$ distance of the identity, where
\begin{equation*}
  m\geq \log_2\log_2 \left(\frac 1{\sqrt{8}\epsilon}\right) + 1. 
\end{equation*}

Combining these and introducing the function $F(U) \coloneqq (f \circ g)^{\circ k}$, we see that if
\begin{equation}
  k\geq \frac{\log_2(1/\eta)}{\log_2(4/3)} + \log_2\log_2 \left(\frac
  1{\sqrt{8}\epsilon}\right)  + 1 \label{eq:k-LB}
\end{equation}
$U$ will be mapped to within $\epsilon$ distance of $\iden$ with probability $\geq 1-\eta$.
Expanding $F(U)$ gives a pulse sequence of the form $R_1 U R_2 \cdots R_n U R_{n+1}$, which can be changed into the form of \eq{eq:format} by conjugating by $R_{n+1}^\dagger$.

The number of pulses ($n$) required for the full refocusing function $F$ is the same as the number of uses of $U$, which is $4^k$.
Thus we see that the number of pulses is bounded by
\begin{equation}
  n=4^k \leq \frac{16}{\eta^5} \log_2^2\left(\frac{1}{\sqrt{8}\epsilon}\right).
\label{eq:1q-twiddle}
\end{equation}
The multiplicative factor of 16 comes from the fact that $k$ may need
to be rounded up to the nearest integer greater than the RHS of \eq{eq:k-LB}. In addition, we have rounded the power of $1/\eta$ up from $2/\log_2(4/3)\approx 4.82$ to $5$.

\textbf{Refocusing in $d$-dimensional systems}.
Though the basic idea of the one-qubit case generalizes to $d$ dimensions, it is more difficult to determine the \jrs, and not at all clear that \rcs\s can even bring arbitrary $d$-dimensional unitary operations close to these \jrs.
However, we will show there exist \jrs\s that can be reached from any unitary.

\textit{Bounding the $d$-dimensional \sr}.---We can generalize the ideas from the one-qubit case to qudits of dimension $d$, with basis $\{\ket{0}, \dots, \ket{d-1}\}$.
The operators acting on the quantum system can be described by the $(d^2-1)$-dimensional Lie algebra $\LieA{su}(d)$ with corresponding Lie group $\LieG{SU}(d)$.
Let $\{\rho_t\}_{t=0}^{d^2-2}$ be a basis for $\LieA{su}(d)$.
It is well-known \cite{pfeifer2003lie} that all $\rho_t$ are traceless and anti-Hermitian.
In addition, let us introduce the unitary Weyl operators \cite{weyl1927quantenmechanik} $\x$ and $\z$ (generalized $d$-dimensional versions of those used in the one-qubit case) by
\begin{equation}
  \x\ket{x}=\ket{x+1 \hspace{0.5em}(\mbox{mod } d)}, \z\ket{x}=\omega^x \ket{x} \label{eq:dd-weyldef}
\end{equation}
where $\omega=\exp(2\pi i/d)$ is a primitive $d$th root of unity.
With $\vect{a}=(a_1, a_2)^T$ and $a_1,a_2\in [d]=\{0,\dots,d-1\}$, we define $\sigma_{\vect{a}}=\z^{a_1}\x^{a_2}$.

We introduce the mapping $f:\LieG{G}\rightarrow\LieG{G}$, defined by
\begin{equation}
  f(U)=\prod_{\vect{a} \in [d]^2} \sigma_{\vect{a}} U \sigma^\dagger_{\vect{a}}. \label{eq:dd-fdef}
\end{equation}

Using the operator norm, we define the \emph{distance} between an operator $W$ and $\iden$ to be $\|W-\iden\|$.
If we define $\varepsilon_0 \coloneqq \|U-\iden\|$, then (see \supp) provided $\varepsilon_0\leq 1/(2\alpha)$, we find that
\begin{equation}
  \varepsilon_m < 2^{-2^m}/\alpha, \label{eq:dd-epsilon_m}
\end{equation} 
where
\begin{equation}
  \alpha=2^{d^2+1}. \label{eq:dd-alphadef}
\end{equation} 
Thus we define the \sr\s to be
\begin{equation}
  \|U-\iden\|=\varepsilon_0\leq 1/(2\alpha).\label{eq:dd-srdef}
\end{equation}

\textit{Finding the $d$-dimensional \jrs}.---We can write $U=e^H$, where $H\in\LieA{su}(d)$.
We show in the \supp~that if $U$ (and hence $H$) is diagonal, $f(U)=\iden$.
Thus the \jrs\s include the neighbourhoods of all diagonal unitaries.

\textit{Bounding the $d$-dimensional \jrs}.---Suppose we have a $W$ such that $f(W)=\iden$, and let $W^\prime=W(\iden+\delta W)$.
Use of the hybrid argument in \cite{vazirani1998power} then yields $\|f(W)-\iden\|\leq d^2\|\delta W\|$ (see \supp~for details).
Recalling \eq{eq:dd-srdef}, we therefore see that if
\begin{equation}
  \|\delta W\|\leq\delta\coloneqq\frac{1}{2 \alpha d^2}, \label{eq:dd-deltaw}
\end{equation}
then $f(W)$ will be in the \sr.

\textit{Bounding the probability of landing in a $d$-dimensional \jr\s after applying a \rc}.---Here we conjugate $U$ with a Haar random unitary $V\in\LieG{SU}(d)$ (i.e.\ uniformly random with respect to the Haar measure \cite{haar1933massbegriff}) and bound the probability that the resulting operator is close to diagonal, and thus in a \jr.
Conjugation is a useful operation to apply since
\begin{equation*}
  \|VUV^\dagger-\iden\|=\|V(U-\iden)V^\dagger\|=\|U-\iden\|
\end{equation*}
and thus, as in the one-qubit case, it leaves the distance from the identity invariant.

We note that there is at least one good choice of $V$: let $V_0$ be a unitary such that $V_0 U V_0^\dagger$ is diagonal.
While $V=V_0$ has zero probability, we argue that there is a non-negligible probability that $V$ will be close to $V_0$.
In the \supp~we show that
\begin{equation}
  \Prob [\|V-V_0\| \leq \delta] \geq (\delta/10)^{d^2-1}. \label{eq:dd-prob}
\end{equation}

\textit{Summary of the $d$-dimensional case}.---We summarize the results below:
\begin{enumerate}
\item
  Given $U\in\LieG{SU}(d)$, the \sr\s is defined (from \eq{eq:dd-srdef}) by $\varepsilon_0\leq 1/(2\alpha)$, where (from \eq{eq:dd-alphadef}) $\alpha=2^{d^2+1}$.
  Within this region, $f$ provides doubly-exponential convergence to $\iden$.
  More specifically, (from \eq{eq:dd-epsilon_m}) we have that $\varepsilon_m < 2^{-2^m}/\alpha$.
\item
  The \jrs\s include $W(\iden+\delta W)$, where $W$ is diagonal, and (from \eq{eq:dd-deltaw}) $\|\delta W\|\leq\delta=1/(2 \alpha d^2)$
\item
  Applying a \rc\s gives us (from \eq{eq:dd-prob}) a probability of at least $p \coloneqq (\delta/10)^{d^2-1}$ of landing in a \jr.
\end{enumerate}
As in the one-qubit case, we now introduce the function $g(U)=VUV^\dagger$, where each application of $g$ chooses a unitary $V$ uniformly at random according to the Haar measure on $\LieG{SU}(d)$.
Consider the function $F(U)=(f\circ g)^{\circ k}$, i.e.\
$f$ and $g$ composed $k$ times.
Following identical logic to the qubit case, we deduce that if
\begin{equation}
  k\geq \frac{\log_2\eta}{\log_2(1-p)} +\log_2 \log_2 \left(\frac{1}{\alpha\epsilon}\right) \label{eq:dd-k-LB}
\end{equation} 
then $U$ will be mapped to within $\epsilon$ distance of $\iden$ with probability $\geq 1-\eta$.
As before, $F$ can then be trivially expanded in the form of \eq{eq:format} to give the required function.

The number of pulses ($n$) required for the full refocusing function $F$ is the same as the number of uses of $U$, which is $d^{2k}$.
Thus we see that the number of pulses looks like
\begin{equation}
  n = d^{2k} \leq d^2 \left(\frac{1}{\eta}\right)^{2^{O(d^4)}} \left(\log_2\left(\frac{1}{\epsilon}\right)-d^2-1\right)^{2 \log_2 d}, \nonumber
\end{equation} 
where the multiplicative factor of $d^2$ comes from the fact that $k$ may need to be rounded up to the nearest integer greater than the RHS of \eq{eq:dd-k-LB}. For fixed $d$, we see that this is similar to \eq{eq:1q-twiddle} from the one-qubit case.
With increasing $d$, we see that the dependence on $\epsilon$ increases only modestly (owing to the decrease in size of the \sr), but the number of steps required to maintain the probability of success, $1-\eta$, increases doubly-exponentially in the Hilbert-space dimension.

Finally, we mention some interesting open questions.
One may ask whether it is possible to have sequences where $\eta=0$.
The randomness is important to our analysis.
Moreover, the function $f$ contains fixed points and cycles of various orders, and the random conjugations serve to break free of these fixed points.
Indeed, we conjecture that there are cycles of all orders.
However, it may be possible to avoid the random conjugations completely.
Numerical simulations strongly suggest these cycles form a zero-measure subset of $\LieG{SU}(d)$, and that the only stable fixed point of $f$ is $\iden$.
We leave rigorous proof of these conjectures as an interesting open problem.

ISBS thanks EPSRC for financial support.
TSC is supported by the Royal Society.
AWH was funded by NSF grants CCF-1111382 and CCF-1452616, ARO contract
W911NF-12-1-0486 and a Leverhulme Trust Visiting Professorship VP2-2013-041.
We are grateful to Jos\'e Figueroa-O'Farrill for code used in Figure~\ref{fig:sphere}.

This work was made possible through the support of grant \#48322 from the John
Templeton Foundation.
The opinions expressed in this publication are those of the authors and do not
necessarily reflect the views of the John Templeton Foundation.

\onecolumngrid
\appendix

\part*{\supp}

\section*{One-Qubit Unitary Noise}

\subsection*{Bounding the \sr}

\supdesc{This section proves the result stated in \eql{eq:sumsq}, that any unitary operation $U\in\LieG{SU}(2)$ can be expressed in the form $U=a\iden+ib\x+ic\y+id\z$, where $a^2+b^2+c^2+d^2=1, \hspace{0.5em}a,b,c,d\in \R$.}

As we made key use of this for qubit refocusing, we recall here the result that any unitary operation $U\in\LieG{SU}(2)$ may be written in the form $U=e^{i \vect{u}\cdot\vect{\sigma}}$, where $\vect{u}=(u_1,u_2,u_3)\in \R^3$, and $\vect{\sigma}=(\x, \y, \z)$.
Since $(\vect{u}\cdot\vect{\sigma})^2=|\vect{u}|^2 \iden$, we see that if $\vect{u}\neq\vect{0}$,
\begin{equation*}
U=e^{i \vect{u}\cdot \vect{\sigma}} = (\cos |\vect{u}|) \iden + i (\sin |\vect{u}|) (\vect{\hat{u}}\cdot \vect{\sigma})
\end{equation*}
where $\vect{\hat{u}}=\vect{u}/|\vect{u}|=(\hat{u}_1,\hat{u}_2,\hat{u}_3)$ is a normalized vector.
Letting $a=\cos|\vect{u}|, b=(\sin|\vect{u}|) \hat{u}_1, c=(\sin|\vect{u}|)
\hat{u}_2,$ and $d=(\sin|\vect{u}|) \hat{u}_3$, we arrive at \eql{eq:sumsq}.

\section*{Refocussing in $d$-dimensional systems}

\subsection*{Properties of $\sigma_{\boldsymbol{a}}$ and $\rho_t$} 

As detailed in the letter, $\{\rho_t\}_{t=0}^{d^2-2}$ is a traceless and anti-Hermitian basis for $\LieA{su}(d)$ (see~\cite{pfeifer2003lie}), and $\{\sigma_{\vect{a}}\}_{\vect{a}\in [d]^2}$ is the group generated by the $d$-dimensional Weyl operators~\cite{weyl1927quantenmechanik}.

Below we list several properties of these operators.

\begin{enumerate}
	\item The $\sigma_{\vect{a}}$'s form an orthogonal basis with respect to the Hilbert-Schmidt inner product\footnote{$\langle A,B \rangle = Tr(A^\dagger B)$~\cite{gohberg1981basic}} for $GL(d,\mathbb{C})$. More specifically, they satisfy $Tr(\sigma_{\vect{a}}^\dagger \sigma_{\vect{b}}) = d\delta_{\vect{a} \vect{b}}$. Note, in addition, that by setting $\vect{a}=\vect{0}$, we have that $Tr(\sigma_{\vect{b}}) = 0$ for $\vect{b} \neq \vect{0}$. \label{prop:sigmabasis}
	\item $\sigma_{\vect{a}}\sigma_{\vect{b}}=\sigma_{\vect{b}}\sigma_{\vect{a}}\omega^{[\vect{a},\vect{b}]}$, where $\omega=\exp(2\pi i/d)$ and $[\vect{a},\vect{b}]$ is the symplectic inner product\footnote{$[\vect{a},\vect{b}]=\vect{a}^T \Omega \vect{b}$, where $\Omega=\smat{0}{1}{-1}{0}$}. \label{prop:sigmacommutator}
	\item $\sum_{\vect{a}\in[d]^2} \omega^{[\vect{a},\vect{b}]}=0$ if $\vect{b}\neq\vect{0}$. \label{prop:sympcommsum}
	\item $\sum_{\vect{a}\in[d]^2} \sigma_{\vect{a}} \sigma_{\vect{b}} \sigma_{\vect{a}}^\dagger = 0$ if $\vect{b} \neq \vect{0}$. This is easily seen by combining Properties~\ref{prop:sigmacommutator} and~\ref{prop:sympcommsum}. \label{prop:sumsigma}
	\item $\sum_{\vect{a}\in[d]^2} \sigma_{\vect{a}} \rho_t \sigma_{\vect{a}}^\dagger = 0 \hspace{0.5em} \forall t$. $\rho_t$ can be expanded in the $\{\sigma_{\vect{a}}\}$ basis (Property~\ref{prop:sigmabasis}). Properties~\ref{prop:sigmabasis} and~\ref{prop:sumsigma}, and the fact that $\rho_t$ is traceless, lead to the result. \label{prop:sumrho}
\end{enumerate}

\subsection*{Bounding the $d$-dimensional \sr}

\supdesc{This section is dedicated to proving \eql{eq:dd-epsilon_m}
	($\varepsilon_m < 2^{-2^m}/\alpha$), with $\alpha$ defined as in \eql{eq:dd-alphadef}
	($\alpha=2^{d^2+1}$).}

For this analysis, we write $U = \iden + \delta U = e^H$, where $H$ is a linear combination of $\rho_t$'s. Furthermore, we impose that $\|\delta U\| \leq 1/2$.

We have that $H=\log(\iden+\delta U)$, and hence the Mercator series\footnote{If $\|A\| < 1$, $\log(\iden+A) = \sum_{k=1}^\infty (-1)^{k+1} A^k/k$} gives us that
\begin{equation}
\|H\|\leq \sum_{k=1}^\infty \frac{\|\delta U\|^k}{k} \leq \sum_{k=1}^\infty \frac{(1/2)^{k-1}}{k}\|\delta U\| < \frac{1}{1-1/2}\|\delta U\| = 2\|\delta U\| \leq 1. \label{eqS:dd-H}
\end{equation}

With $f$ as defined in \eql{eq:dd-fdef} and writing $U=\iden + \delta U$, we see that
\begin{equation*}
f(U)=\prod_{\vect{a}\in[d]^2} \sigma_{\vect{a}} (\iden + \delta U) \sigma_{\vect{a}}^\dagger
= \iden + \sum_{\vect{a}\in[d]^2} \sigma_{\vect{a}} \delta U \sigma_{\vect{a}}^\dagger + \sum_{\vect{a} < \vect{b}} \sigma_{\vect{a}} \delta U \sigma_{\vect{a}}^\dagger \cdot \sigma_{\vect{b}}
\delta U \sigma_{\vect{b}}^\dagger + \cdots, 
\end{equation*}
where $\vect{a}<\vect{b} \Leftrightarrow da_1+a_2 < db_1+b_2$, for $\vect{a}=(a_1,a_2), \vect{b}=(b_1,b_2) \in [d]^2$. 

After moving the $\iden$ to the left-hand side, we take the operator norm\footnote{$\|A\|=\sup_{\ket{\psi}\in\mathbb{C}^d, |\ket{\psi}|=1}\|A\ket{\psi}\|$} of both sides and use the triangle inequality and sub-multiplicative property\footnote{$\|AB\| \leq \|A\| \|B\| \hspace{1em} \forall A,B\in\LieG{SU}(d)$} to deduce that
\begin{equation}
\|f(U)-\iden\|
\leq\left\| \sum_{\vect{a} \in [d]^2} \sigma_{\vect{a}} \delta U \sigma_{\vect{a}}^\dagger\right\| + \comb{d^2}{2} \|\delta U\|^2 + \comb{d^2}{3} \|\delta U\|^3 + \cdots
=\left\| \sum_{\vect{a} \in [d]^2} \sigma_{\vect{a}} \delta U \sigma_{\vect{a}}^\dagger\right\| + \sum_{j=2}^{d^2} \comb{d^2}{j} \|\delta U\|^j. \label{eqS:dd-sr10}
\end{equation}

Since $\|\delta U\| \leq 1/2$, $\|\delta U\|^j \leq \|\delta U\|^2$ for $j \geq 2$.
Thus
\begin{equation}
\|f(U)-\iden\|
\leq \left\| \sum_{\vect{a} \in [d]^2} \sigma_{\vect{a}} \delta U \sigma_{\vect{a}}^\dagger\right\| + \left[\comb{d^2}{2}+\cdots+\comb{d^2}{d^2}\right]\|
\delta U\|^2
= \left\| \sum_{\vect{a} \in [d]^2} \sigma_{\vect{a}} \delta U \sigma_{\vect{a}}^\dagger\right\| + (2^{d^2}-d^2-1)\|\delta U\|^2. \nonumber
\end{equation}

Now consider the first term on the right-hand side, and note that we can write $U$ in the form $U=e^H$, where $H$ is a linear combination of $\rho_t$'s.
Hence
\begin{align}
\left\| \sum_{\vect{a} \in [d]^2} \sigma_{\vect{a}} \delta U \sigma_{\vect{a}}^\dagger\right\|
&= \left\| \sum_{\vect{a} \in [d]^2} \sigma_{\vect{a}} (e^H-\iden) \sigma_{\vect{a}}^\dagger\right\| = \left\| \sum_{\vect{a} \in [d]^2} \sigma_{\vect{a}} \left(\sum_{k=1}^\infty \frac{H^k}{k!}\right) \sigma_{\vect{a}}^\dagger\right\| \nonumber\\
&\leq \left\| \sum_{\vect{a} \in [d]^2} \sigma_{\vect{a}} H \sigma_{\vect{a}}^\dagger\right\| + \sum_{k=2}^\infty \frac{1}{k!} \left\| \sum_{\vect{a} \in [d]^2} \sigma_{\vect{a}} H^k \sigma_{\vect{a}}^\dagger\right\| \nonumber \\
&\leq \left(\sum_{k=2}^\infty\frac{1}{k!}\right) \|H\|^2 d^2 \nonumber \\
&\leq 4d^2(e-2)\|\delta U\|^2. \label{eqS:dd-sr30}
\end{align}
where the second line follows from the triangle inequality.
The first term in the second line is 0 by Property~\ref{prop:sumrho}.
The third line then follows by the triangle inequality, the sub-multiplicative property, and the fact that $\|H\| < 1$ (from \eqs{eqS:dd-H}). The final line follows from \eqs{eqS:dd-H}.
Hence we discover that
\begin{equation}
\|f(U)-\iden\| \leq (2^{d^2}+d^2(4e-9)-1) \|\delta U\|^2 < \alpha \|\delta U\|^2 \label{eqS:dd-iterativeshrink} \vspace{1ex}
\end{equation}
with $\alpha$ defined as in \eql{eq:dd-alphadef}.
This leads, as described in the Letter, to the \sr\s being defined to be
\begin{equation}
\|U-\iden\|=\epsilon_0\leq 1/(2\alpha). \label{eqS:muoveralphaless1/2}
\end{equation}

In addition, if, as in the qubit case, $\varepsilon_m$ is the distance from $\iden$ after $m$ applications of $f$, then the repeated application of \eqs{eqS:dd-iterativeshrink} implies \eql{eq:dd-epsilon_m}.

\subsection*{Finding the $d$-dimensional \jrs}

\supdesc{In this section we show that $f(U)=\iden$ if $U$ is diagonal.}

Property~\ref{prop:sigmabasis} allows us to write
\begin{equation}
H=\sum_{\vect{a}\in[d]^2, \vect{a}\neq\vect{0}} \lambda_{\vect{a}} \sigma_{\vect{a}}, \label{eqS:Hdef}
\end{equation}
where $\lambda_{\vect{a}}\in \mathbb{C}\hspace{0.5em}\forall \vect{a}$, and $\vect{a}=0$ is excluded from the sum because $H\in\LieA{su}(d)$ is traceless. In addition, if $U$ is diagonal, we have that $H$ is diagonal, and thus the only non-zero $\lambda_{\vect{a}}$'s are those corresponding to diagonal $\sigma_{\vect{a}}$'s (i.e. $\vect{a}=(a_1,0)^T$).

From \eql{eq:dd-fdef}, we have
\begin{equation}
f(U)=\prod_{\vect{c}\in[d]^2}\sigma_{\vect{c}}U\sigma_{\vect{c}}^\dagger=\prod_{\vect{c}\in[d]^2}\exp \Lambda_{\vect{c}} \label{eq:dd-fwithlambda}
\end{equation}
where
\begin{equation*}
\Lambda_{\vect{c}}=\sigma_{\vect{c}}H\sigma_{\vect{c}}^\dagger=\sum_{\vect{a}\neq\vect{0}} \omega^{[\vect{c},\vect{a}]} \lambda_{\vect{a}} \sigma_{\vect{a}},
\end{equation*}
in which we have used \eqs{eqS:Hdef} and Property~\ref{prop:sigmacommutator} to deduce the final equality. Note that the non-zero terms of the sum are diagonal, and hence all $\Lambda_{\vect{c}}$ commute. Thus using Property~\ref{prop:sympcommsum}, we see that
\begin{equation*}
f(U)=\exp\left(\sum_{\vect{c}}\Lambda_{\vect{c}}\right)
=\exp\Bigg( \sum_{\vect{a}\neq\vect{0}} \underbrace{\left(\sum_{\vect{c}}\omega^{[\vect{c},\vect{a}]}\right)}_{=0} \lambda_{\vect{a}} \sigma_{\vect{a}} \Bigg) = \iden. \nonumber
\end{equation*}

\subsection*{Bounding the size of the \jrs}

\supdesc{This section provides the proof of the bound given in \eql{eq:dd-deltaw}, which states that if an operator is within a distance $1/(2\alpha d^2)$ from a diagonal operator, it will be mapped into the \sr.}

$f(U)$ is a product of operators, containing $d^2$ instances of $U$.
The hybrid argument in \cite{vazirani1998power} then implies that
\begin{equation}
\|f(U)-f(V)\| \leq d^2 \|U-V\|. \label{eq:hybrid}
\end{equation}
Suppose that we have a $W$ such that $f(W)=\iden$, and define $W^\prime=W(\iden+\delta W)$.
\eqs{eq:hybrid} then gives
\begin{equation*}
\|f(W^\prime)-\iden\| \leq d^2\|\delta W\|.
\end{equation*}
Thus to ensure that $f(W^\prime)$ is in the \sr, we must have that
\begin{equation*}
\|\delta W\| \leq \delta \coloneqq \frac{1/(2\alpha)}{d^2}=\frac{1}{2 \alpha d^2}
\end{equation*}
as described in \eql{eq:dd-deltaw}.

\subsection*{Bounding the probability of landing in a $d$-dimensional jumping region after applying a \rc}

\supdesc{This section is dedicated to proving \eql{eq:dd-prob}, which states that a random conjugation has probability $\geq \delta^{O(d^2)}$ of sending a given matrix to a matrix within $\delta$ of being diagonal.}

As described in the Letter, we choose a unitary operator $V\in\LieG{SU}(d)$ uniformly at random according to the Haar measure \cite{haar1933massbegriff}, and lower-bound the probability that it is close to $V_0$, where $V_0\in\LieG{SU}(d)$ and $V_0 U V_0^\dagger$ is diagonal.

We first note that $\Prob [\|V-V_0\| \leq \delta]$ is independent of $V_0$, and so wlog we consider $V_0=\iden$.
Consider the map $\exp:\LieA{su}(d)\rightarrow \LieG{SU}(d)$, and let
$B_r=\{ s\in\LieA{su}(d) : \|s\| \leq r\}$, for $r\leq \pi$.
Note that $\exp(B_r)$ is a ball around $\iden$ in $\LieG{SU}(d)$ of
radius $|\exp(ir)-1| = 2\sin(r/2)$.
Thus the pre-image of the ball of
radius $\delta$ is $B_\nu$ with
\begin{equation}
\nu = 2 \arcsin(\delta/2)
\label{eq:deltanu}
\end{equation}
Now, the volume of $B_r$ is $\vol(B_r)=c r^{d^2-1}$, where $c$ is dependent upon $d$, and we are using the Euclidean metric on $\LieA{su}(d)$.

Lemma~4 in \cite{szarek1997metric} provides the result
\begin{equation*}
\frac{4}{10} \|s\| \leq \|\exp(s)-\iden\| \leq \|s\|
\end{equation*}
for $s\in \LieA{su}(d)$; the upper bound holds for all $s$, and the lower bound holds for $\|s\|\leq \pi/4$.
Thus
\begin{itemize}
	\item if $\nu \leq \pi/4$, then $\vol(\exp(B_\nu)) \geq \frac{4}{10} \vol(B_\nu) = \frac{4}{10} c \nu^{d^2-1}$; and
	\item since $\exp(B_\pi)=\LieG{SU}(d)$, we have that $\vol(\LieG{SU}(d)) \leq \vol(B_\pi) = c \pi^{d^2-1}$.
\end{itemize}
Hence the probability that a random operator $V \in \LieG{SU}(d)$ is within distance $\delta$ from $\iden$ (or any other $V_0$) is lower bounded by
\begin{equation*}
\Prob[\|V-V_0\| \leq \delta] \geq \frac{4}{10} \left(\frac{\nu}{\pi}\right)^{d^2-1}.
\end{equation*}

In addition, \eqs{eq:deltanu} implies that
\begin{equation*}
\nu = 2 \arcsin(\delta/2) \geq \delta,
\end{equation*}
hence we arrive at
\begin{equation*}
\Prob[\|V-V_0\| \leq \delta]
\geq \frac{4}{10} \left(\frac{\delta}{\pi}\right)^{d^2-1} \geq \left(\frac{\delta}{10}\right)^{d^2-1}
\end{equation*}
as given in \eql{eq:dd-prob}.

\section*{Solovay-Kitaev without inverses}

\supdesc{This section gives a full proof of the inverse-free Solovay-Kitaev theorem.}

The standard Solovay-Kitaev theorem~\cite{nielsen2010quantum} states:
\begin{thm}[Solovay-Kitaev]\label{S-K-thm}
	Let $\mathcal{G}$ be a universal quantum gate set, and let $\mathcal{G}^\dagger \coloneqq \{ V^\dagger : V\in\mathcal{G} \}$.
	For any $\epsilon > 0$ and any $U\in\LieG{SU}(d)$, there is an efficient classical algorithm that constructs a sequence of gates $V_L \cdots V_1 V_0$ with $V_i\in\mathcal{G} \cup \mathcal{G}^\dagger$ and $L=\polylog(1/\epsilon)$ such that $\|V_L \cdots V_1 V_0 - U\| \leq \epsilon$.
\end{thm}
(\noindent Note that the norm used in~\cite{nielsen2010quantum} is the trace norm, whereas we are using the operator norm. But these are equivalent up to an unimportant factor of~2.)

The following is the key lemma, using part of our refocusing result to show that inverses can be approximated efficiently:
\begin{lem}\label{sk-lem2}
	Let $\Delta$ be a $(\mu/\alpha)$-net for $\LieG{SU}(d)$ for constant $\mu < 1$ and $\alpha=2^{d^2}+d^2(4e-9)-1$.
	Let $\mathcal{W}$ be the $d$-dimensional Weyl operators.
	For any $\epsilon$ and any $U\in \LieG{SU}(d)$, there is an efficient classical algorithm that constructs a product of unitary operators $g_\epsilon(U)$ from the set $\Delta\cup\mathcal{W}$, of length $\polylog(1/\epsilon)$, for which $\|g_\epsilon(U)-U^\dagger\| = O(\epsilon)$.
\end{lem}

\begin{proof}
	Since $\Delta$ is a $(\mu/\alpha)$-net for $\LieG{SU}(d)$, there exists $W\in\Delta$ with $\|U^\dagger-W\|\leq \mu/\alpha$, hence $\|\iden-WU\|\leq \mu/\alpha$.
	Thus $WU$ is in the \sr.
	Let $f$ be the mapping $f:\LieG{G}\rightarrow\LieG{G}$ defined in \eql{eq:dd-fdef} (which can manifestly be computed efficiently).
	By \eql{eq:dd-epsilon_m}, $\|f^m(WU) - \iden \| \leq \mu^{2^m}/\alpha$.
	Setting $m=O(\log\log(1/\epsilon))$, we have \mbox{$\|f^m(WU) - \iden \| = O(\epsilon)$}.
	
	Now, $f^m(WU)$ is a sequence of unitary operators of the form $R_1 WU R_2 WU \cdots R_{L-1} WU R_L WU$ (where the $R_i$ are Weyl operators).
	By removing the trailing $U$ from this sequence to form the sequence $g_\epsilon(U) =\linebreak[4] R_1 WU R_2 WU \cdots WU R_L W$, we have $\|g_\epsilon(U)-U^\dagger\| = O(\epsilon)$ by unitary invariance of the norm.
	$f^m(WU)$ has length $3\times 4^m = \polylog(1/\epsilon)$, hence $g_\epsilon(U)$ also has length $\polylog(1/\epsilon)$.
\end{proof}

Putting Theorem~\ref{S-K-thm} and Lemma~\ref{sk-lem2} together, we obtain the inverse-free Solovay-Kitaev theorem:
\begin{thm}[Inverse-free Solovay-Kitaev]
	Let $\mathcal{G}$ be a universal quantum gate set---a finite set of
	elements in $SU(d)$ such that $\langle \mathcal{G} \rangle$ is dense in $SU(d)$---that contains the Weyl operators.
	For any $\epsilon > 0$ and given any $U\in\LieG{SU}(d)$, there is an efficient classical algorithm that constructs a sequence of gates $V_L \cdots V_1 V_0$ with $V_i\in\mathcal{G}$ and $L=\polylog(1/\epsilon)$ such that \mbox{$\|V_L \cdots V_1 V_0 - U\| \leq \epsilon$}.
\end{thm}
\begin{proof}
	We wish to apply Lemma~\ref{sk-lem2} to $V^\dagger\in \mathcal{G}^\dagger$.
	Since $\mu/\alpha$ is constant, we can generate a $(\mu/\alpha)$-net, denoted $\Delta$, from constant-length products of operators from $\mathcal{G}$. One can see that constant-length products are sufficient as follows.
	Given a set of unitary operators $\mathcal{U}=\{U_1, ..., U_N\}$, let us define $w(\mathcal{U}) \coloneqq \max_{V \in \LieG{SU}(2)} \min_{U \in \mathcal{U}} \|V - U\|$.
	Let us further define $v(L) \coloneqq w(\{\mbox{set of all products of operators from $\mathcal{G}$ of length $L$}\})$.
	Thus clearly $v(L) \leq v(L-1)$.
	Also, since $\langle G\rangle$ is dense in $\LieG{SU}(d)$, $\lim_{L\to\infty} v(L) = 0$.
	In other words: for all $\delta > 0$ there exists an $L$ such that $v(L) < \delta$.
	
	Furthermore, since $\mathcal{G}$ contains the Weyl operators, Lemma~\ref{sk-lem2} allows us to construct a $\polylog(1/\epsilon)$-length product $g_\epsilon(V)$ of operators from $\mathcal{G}$ such that $\|g_\epsilon(V) - V^\dagger\| \leq C \epsilon$, for some constant~$C$.
	
	Theorem~\ref{S-K-thm} lets us construct a product of gates $V_L \cdots V_1 V_0$ with $V_i\in\mathcal{G} \cup \mathcal{G}^\dagger$ and $L=\polylog(1/\epsilon)$ such that $\|V_L \cdots V_1 V_0 - U\| \leq \epsilon/2$, where $\mathcal{G}^\dagger \coloneqq \{ V^\dagger : V\in\mathcal{G} \}$.
	We construct a new product of gates $V^\prime_{L^\prime} \cdots V^\prime_1 V^\prime_0$ with $V^\prime_i\in\mathcal{G}$, by replacing each $V_i\in \mathcal{G}^\dagger\setminus\mathcal{G}$ with $g_{\epsilon/(2LC)}(V_i^\dagger)$ (where $V_i^\dagger \in \mathcal{G}$).
	Hence $\|V^\prime_{L^\prime} \cdots V^\prime_1 V^\prime_0 - V_L \cdots V_1 V_0\| \leq \epsilon/2$.
	Then $\|V^\prime_{L^\prime} \cdots V^\prime_1 V^\prime_0 - U\| \leq \epsilon$.
	Since we have replaced at most $L=\polylog(1/\epsilon)$ gates, we see that $L^\prime = L \cdot \polylog(2LC/\epsilon) = \polylog(1/\epsilon)$.
\end{proof}

\nocite{apsrev41Control}
\bibliographystyle{apsrev4-1}
\bibliography{references}

\end{document}